\documentclass[english,english,pra,english,preprint,amsmath,amssymb,aps,longbibliography,showkeys, titlepage]{revtex4-2}
\usepackage{lmodern}
\usepackage{pgf}
\usepackage{mathtools}
\usepackage[T1]{fontenc}
\usepackage[latin9]{inputenc}
\usepackage{amsmath}
\usepackage{amsthm}
\usepackage{amssymb}
\usepackage{graphicx}
\usepackage{dirtytalk}
\usepackage[section]{placeins}
\newtheorem*{theorem*}{Theorem}

\makeatletter
\usepackage[T1]{fontenc}
\usepackage{amsmath}
\usepackage{amsthm}
\usepackage{amssymb}
\usepackage{graphicx}

\newtheorem{theorem}{Theorem}

\usepackage{graphicx}
\usepackage{float}
\usepackage{xcolor}
\usepackage[hypertexnames=false]{hyperref}
\hypersetup{
	breaklinks = true,
    colorlinks = true,
    citecolor = {blue},
	urlcolor = {blue},
	linkcolor = {blue}
}
\usepackage{babel}
\makeatother
\begin{document}
\date{\today}
\title{On Minimizing Phase Space Energies}
\author{Michael Updike}
\email{michaelupdike@princeton.edu}

\affiliation{Princeton Plasma Physics Laboratory, Princeton University, Princeton,
NJ 08540}
\affiliation{Department of Astrophysical Sciences, Princeton University, Princeton,
NJ 08540}
\author{Nicholas Bohlsen}
\email{nbohlsen@pppl.gov}

\affiliation{Princeton Plasma Physics Laboratory, Princeton University, Princeton,
NJ 08540}
\affiliation{Department of Astrophysical Sciences, Princeton University, Princeton,
NJ 08540}
\author{Hong Qin}
\email{hongqin@princeton.edu }

\affiliation{Princeton Plasma Physics Laboratory, Princeton University, Princeton,
NJ 08540}
\affiliation{Department of Astrophysical Sciences, Princeton University, Princeton,
NJ 08540}
\author{Nathaniel J. Fisch}
\email{fisch@princeton.edu }

\affiliation{Princeton Plasma Physics Laboratory, Princeton University, Princeton,
NJ 08540}
\affiliation{Department of Astrophysical Sciences, Princeton University, Princeton,
NJ 08540}
\begin{abstract}
A primary technical challenge for harnessing fusion energy is to control and extract energy from a non-thermal distribution of charged particles. The fact that phase space  evolves by symplectomorphisms fundamentally limits how a distribution may be manipulated. While the constraint of phase-space volume preservation is well understood, other constraints remain to be fully appreciated. To better understand these constraints, we study the problem of extracting energy from a distribution of particles using area-preserving and symplectic linear maps. When a quadratic potential is imposed, we find that the maximal extractable energy can be computed as trace minimization problems. We solve these problems and show that the extractable energy under linear symplectomorphisms may be much smaller than the extractable energy under special linear maps. The method introduced in the present study enables an energy-based proof of the linear Gromov non-squeezing theorem. 
\end{abstract}
\maketitle

\section{Introduction}
As the deployment of commercial fusion energy accelerates, it becomes increasingly indispensable to develop phase space engineering techniques to energize fusing particles and extract energy from fusion product particles.  For standard magnetic confinement D-T fusion, self-sustained burning requires that the energy of alpha particles be transferred to fusing ions, which remains a technical challenge \citep{Budny2018}. Phase space engineering techniques have been designed to channel the alpha particle energy to fusing ions directly via electromagnetic waves \citep{Fisch1992, Fisch1995a, Fisch1995, Herrmann1997,Ochs2015}. For advanced fuel fusion using p-B11 or D-He3, the fusion energy released is carried by charged particles \citep{Bussard1994,Rostoker1997,Nevins1998,Volosov2006,Ochs2021,Kolmes2022a,Ochs2022,Munirov2023,Mlodik2023,Ochs2024}. It is thus possible to convert the fusion energy directly into electricity by manipulating the charged particles of fusion products using electromagnetic fields. On the other hand, advanced aneutronic fusion operates in non-thermalized conditions, necessitating substantial power circulation to maintain fusing particles in non-equilibrium energy states \citep{Nevins1998,Nevins1998a,Qin2024,Qin2025}. For this purpose, energy needs to be extracted from the thermalized fusing particles and converted into kinetic energy.

Electromagnetic manipulation of charged particles encounters inherent physical constraints. Namely, the particle phase space must evolve by Hamiltonian symplectomorphisms. Liouville's theorem establishes that while phase space volume occupied by particles can be reconfigured, it cannot be reduced. This and other constraints led Qin et al. \citep{Qin2025, Qin2024} to frame a pivotal inquiry: How much energy can we electromagnetically harvest from fusion byproducts (like alpha particles from  p-B11, D-He3, or D-T  reactions)? Conversely, what is the lowest attainable energy configuration through electromagnetic processes? Phase-space conservation ensures this baseline state remains above zero, effectively capping the potential energy extraction achievable via radiofrequency waves in plasma systems. 

To give rigorous estimates on these questions, we study how much energy can be extracted from a particle distribution undergoing Hamiltonian time evolution. To be general, let $\mathcal{P}$ be a particle phase space, mathematically a symplectic manifold. A collection of many particles can be described by a distribution function $f_{t}:\mathcal{P} \to [0,\infty)$, which gives the number of particles in a region $\mathcal{V}$ at time $t$ as $\int_{\mathcal{V}} f_{t}$. 
Given a time-dependent Hamiltonian $\mathcal{H}(t): \mathcal{P} \to \mathbb{R}$ generating a complete flow $\phi_{t}: \mathcal{P} \to \mathcal{P}$, the distribution function evolves in the absence of collisions as $f_{t} = f_{0} \circ \phi_{t}^{-1}$. The Hamiltonian $\mathcal{H}(t)$ can be taken to include the self-generated fields of the particles, as well as any externally applied fields, in which case the evolution relation $f_{t} = f_{0} \circ \phi_{t}^{-1}$ is equivalent to the Vlasov equation. Given a reference energy function $\mathcal{E}: \mathcal{P} \to \mathbb{R}$, representing the energy of particles in the absence of fluctuating fields, we define the energy of a particle distribution as $E[f_t] = \int_{\mathcal{P}} f_{t}$. Our question then translates to finding $\inf_{\phi \in \text{Ham}(\mathcal{P})} E[f_{0} \circ \phi^{-1}]$, the minimal energy a particle distribution must maintain under Hamiltonian time evolution. 

When the allowed transforms $\{\phi\}$ are relaxed to be merely invertible and area-preserving, we recover the problem posed by Gardner in \citep{Gardner}. Under these relaxed assumptions, and with enough decay assumptions on $f_{0}$, the Gardner energy $E_{G} \coloneqq \inf_{\{\phi \}} E[f_{0} \circ \phi^{-1}]$ can be computed by sequentially permuting equal-measure sets in phase space. This procedure has come to be known as Gardner's restacking algorithm \citep{Kolmes2020,Kolmes2024,Kolmes2022,Dodin2005}.  While each permutation in Gardner restacking is noncontinuous, smooth approximations may be found using the theory of Dacogna and Moser \citep{dac}. Thus $E_{G}$ defines the minimal energy even when $\{\phi\}$ are restricted to be area-preserving diffeomorphisms.

When $\{\phi \}$ are further restricted to be symplectomorphisms,  one must ponder whether symplectic maps behave rigidly or flexibly with respect to the problem at hand. Unlike area-preserving maps, symplectic maps can quite restrictive. Gromov's nonsqueezing theorem \citep{Gromov, Hofer} states there is no symplectic embedding of the ball $B_{2n}(r) = \{ \mathbf{z} \in \mathbb{R}^{2n}: |\mathbf{z}|^2 < r^2 \}$ into the symplectic cylinder $\mathcal{Z}(R) = \{(\mathbf{x}, \mathbf{p}) \in \mathbb{R}^{2n}: x_{1}^2 + p_{1}^2 < R^2 \}$ except when $r\leq R$. Yet symplectic maps can also be quite flexible owing to Darboux's theorem, which implies there are no local invariants of symplectic manifolds \citep{Darboux1882, Mcduff}. This flexibility versus rigidity conundrum can be answered using a result of Katok \citep[p. 545]{Katok}. Namely, for any two equal-measure, compact subsets $A, B$ of a symplectic manifold, and connected open set $A \cup B \subset U$, there is a Hamiltonian symplectomorphism $\psi$ supported in $U$ making the symmetric difference between $\psi^{-1}(B)$ and $A$ arbitrarily small. This result allows one to approximate Gardner restacking with Hamiltonian symplectomorphisms, thereby showing $E_{G} = \inf_{\phi \in \text{Ham}(\mathcal{P})}E[\phi]$. The details of this are given in \ref{Gromov=Gardner}.

While we have answered our original question, we have done so unsatisfactorily. The symplectic transformations that take an initial distribution $f_{0}$ close to its minimal energy must generically include large gradients, and thus be infeasible to implement physically. Consider, for example, the problem of embedding all but an $\epsilon$ amount of $B_{2n}(r)$ into the cylinder $\mathcal{Z}_{2n}(R)$. It was shown by Sackel, et. al. \cite[p.1116]{qualifications} that for any fixed $r>R$, there is a positive constant $C$ such that the Lipchitz constant $L(\psi)$ of any symplectic embedding $\psi: B^{2n}(r) \to \mathbb{R}^{4}$ must satisfy $L(\psi)^2 \geq C \epsilon^{-1}$. 

To remedy this large gradients problem, we must look for the infimum of $E[f_{0} \circ \phi^{-1}]$ over a more suitable family of symplectomorphisms. The simplest case, and the one we study in this paper, is that $f_{0}$ is supported on the smallest scale we can manipulate. Such a situation arises in accelerator and plasma physics when one considers beams of particles \citep{Davidson01-all, Qin2011,Qin2010}. In such a case, we may safely neglect the topology of $\mathcal{P}$ and assume $\mathcal{P} = \mathbb{R}^{2n}$ with its standard symplectic structure. We may also approximate any allowed symplectomorphism by its linearization $\phi(\mathbf{z}) \approx A (\mathbf{z} - \mathbf{b})$ where $A \in Sp(2n)$ is a symplectic matrix and $\mathbf{b} \in \mathbb{R}^{2n}$. This amounts to ignoring cubic and higher-order terms in the Hamiltonian (see \ref{Ham}).

In this approximation, our refined question becomes to find 
\begin{equation}\label{Esp}
E_{Sp(2n)} \coloneqq \inf_{\mathbf{b} \in \mathbb{R}^{2n}, \; A \in Sp(2n)} \int \mathcal{E}(\mathbf{z}) f_{0}(A^{-1} \mathbf{z} + \mathbf{b}) d^{2n}\mathbf{z},  
\end{equation}
a quantity we call the linear Gromov energy. Since $f_{0}$ is assumed to be supported on the smallest scale we can manipulate, it is reasonable to make the further assumption that $\mathcal{E}(\mathbf{z})$ is well-approximated by a quadratic polynomial in Eq.\,(\ref{Esp}). To explore the constraints of symplectomorphisms, we will also compute 
\begin{equation}E_{SL(2n)} \coloneqq  \inf_{\mathbf{b} \in \mathbb{R}^{2n}, \; A \in SL(2n)} \int \mathcal{E}(\mathbf{z}) f_{0}(A^{-1} \mathbf{z} + \mathbf{b}) d^{2n}\mathbf{z},\end{equation}
which we refer to as the linear Gardner energy. Since every symplectic matrix is area-preserving, we will have that $E_{SL(2n)} \leq E_{Sp(2n)}$, but we should not generally expect equality when $n > 1$. Indeed, we will show this to be the case below.

\section{Analysis} \label{main}
We now answer the question set forth. We will assume that $\mathcal{E}$ is a quadratic polynomial of the form $\mathcal{E}(\mathbf{z}) = a + \mathbf{b} \cdot \mathbf{z} + \mathbf{z}^{T} V \mathbf{z}$ where $V$ is a symmetric matrix. To ensure $\mathcal{E}(\mathbf{z})$  has a unique minimum, we assume that $V$ is positive definite. We will show in \ref{appendix} how positive semi-definite $V$ can be treated. Since $V$ is invertible, there is a constant vector $\mathbf{d}$ and a scalar $V_{0}$ such that $\mathcal{E}(\mathbf{z}) = V_{0} + (\mathbf{z} - \mathbf{d})^{T} V (\mathbf{z} - \mathbf{d})$. For any nonnegative measurable function $f_{0}$ such that $\int |\mathbf{z}|^2 f_{0} < \infty$ we define the moments 
\begin{align}
    N & \coloneqq \int f_{0}(\mathbf{z}) d^{2n}\mathbf{z}, \\ 
\mathbf{c} & \coloneqq \frac{1}{N} \int \mathbf{z} f_{0}(\mathbf{z}) d^{2n} \mathbf{z}, \nonumber \\
    H & \coloneqq \int (\mathbf{z}\otimes \mathbf{z}) f_{0}(\mathbf{z}+\mathbf{c}) d^{2n} \mathbf{z}. \nonumber
\end{align}

These moments can be computed either directly or by differentiating the Fourier transform of $f_{0}$. It is important to note that $H$ is a positive definite matrix since for any $\xi \neq 0 \in \mathbb{R}^{2n}$, $\xi^{T} H \xi = \int (\xi \cdot \mathbf{z})^2 f_{0}(\mathbf{z} + \mathbf{c})d^{2n}\mathbf{z} > 0$. For any fixed $A \in SL(2n)$, we have that 
\begin{align} \label{affine}
    \inf_{\mathbf{b} \in \mathbb{R}^{2n}} \int \mathcal{E}(\mathbf{z}) f_{0}(A^{-1}\mathbf{z} + \mathbf{b}) &= \inf_{\mathbf{b} \in \mathbb{R}^{2n}} \int \mathcal{E}(\mathbf{z})\,  f_{0}(A^{-1}(\mathbf{z} - \mathbf{d} -\mathbf{b})  + \mathbf{c}) d^{2n}\mathbf{z} \\
    & = V_{0} N + \inf_{\mathbf{b} \in \mathbb{R}^{2n}} \int (\mathbf{z} + \mathbf{b})^{T} V(\mathbf{z} + \mathbf{b}) f_{0}(A^{-1}\mathbf{z} + \mathbf{c})d^{2n}\mathbf{z} \nonumber \\
    & =V_{0} N + \int \mathbf{z}^{T} V \mathbf{z} \, f_{0}(A^{-1}\mathbf{z} + \mathbf{c})d^{2n}\mathbf{z} + N \inf_{\mathbf{b} \in \mathbb{R}^{2n}} \mathbf{b}^{T} V \mathbf{b} \nonumber \\
    & = V_{0} N + \int \mathbf{z}^{T} V \mathbf{z} \,f_{0}(A^{-1}\mathbf{z} + \mathbf{c}) d^{2n}\mathbf{z},  \nonumber
\end{align}
where we have used that $\int \mathbf{z} f_{0}(\mathbf{A}^{-1}\mathbf{z} + \mathbf{c})d^{2n}\mathbf{z}= \mathbf{A}\int(\mathbf{z}'-\mathbf{c})f_{0}(\mathbf{z}')d^{2n}\mathbf{z}'=0$ and  $\mathbf{b}^{T} V \mathbf{b} \geq 0$. 
Eq.\,(\ref{affine}) can be paraphrased as saying that the optimal $\mathbf{b}$ is such that the center of mass of $f_{0}$ lies at the potential minimum. We now must find the optimal $A$. We start with the simplification
\begin{equation}
    \int \mathbf{z}^{T} V \mathbf{z}\, f_{0}(A^{-1} \mathbf{z} + \mathbf{c}) d^{2n}\mathbf{z} = \text{tr}(V A H A^{T}) . 
\end{equation}
Thus to compute either $E_{SL(2n)}$ or $E_{Sp(2n)}$ we must solve a trace minimization problem. 
\subsection{Linear Gardner Energy} 
We first analyze the case that $A \in SL(2n)$ since the linear algebra is more familiar. Given any symmetric matrix $M$ there exists a special orthogonal,
hence $SL(2n)$, matrix $O_{M}$ such that $O_{M}^{T} MO_{M} = D_{M}$, where $D_{M}$ is diagonal. We may write any $A \in SL(2n)$ as $A = O_{V} B O_{H}^{T}$ with $B \in SL(2n)$. Our trace to be minimized then becomes
\begin{equation}
\text{tr}(VA H A^{T}) = \text{tr}(D_{V} B D_{H} B^{T}) = \text{tr}((D_{H}^{1/2} B^{T} D_{V}^{1/2}) (D_{H}^{1/2} B^{T} D_{V}^{1/2})^{T}),
\end{equation}
where we have used that $D_{H}$ and $D_{V}$ are positive definite in taking their roots. Since $X\coloneqq (D_{H}^{1/2} B^{T} D_{V}^{1/2}) (D_{H}^{1/2} B^{T} D_{V}^{1/2})^{T}$ is positive definite, $X$ has positive eigenvalues so we may apply the AM-GM inequality to the eigenvalues of $X$ to derive that
\begin{equation}\label{eqn9}
\text{tr}(X) \geq 2n\,\text{det}(X)^{1/2n} = 2n \, \text{det}(VH)^{1/2n},
\end{equation}
where we used $\text{det}(B) = 1$. Equality in Eq.\,(\ref{eqn9}) is obtainable iff \\ $B = \text{det}(HV)^{1/4n} D_{V}^{-1/2}O D_{H}^{-1/2}$ for some $O \in SO(2n)$. This allows us to conclude that
\begin{equation} \label{gardner}
E_{SL(2n)} = NV_{0} + 2n \, \text{det}(VH)^{1/2n}. 
\end{equation}
\subsection{Linear Gromov Energy}
We now restrict to the trickier case that $A \in Sp(2n)$. Williamson's theorem \citep{williamson} states that for any symmetric, positive-definite matrix $M$ there exists a symplectic matrix $S_{M}$ such that 
\begin{equation}
    S_M^{T} M S_M = \begin{bmatrix} \mathcal{D}_{M} & 0 \\ 0 &  \mathcal{D}_{M} \end{bmatrix} = \mathcal{D}_{M} \oplus \mathcal{D}_{M},
\end{equation}
where $\mathcal{D}_{M} = \text{diag}(\lambda^M_{1},\hdots,\lambda^M_{n})$ with $\lambda^M_{1} \geq \hdots \geq \lambda^M_{n} > 0$. The values $\lambda^M_{i}$ are called the symplectic eigenvalues of $M$. For computational purposes, we note that $\lambda^M_{i}$ is a symplectic eigenvalue of $M$ iff $\pm i \lambda^M_{i}$ are eigenvalues of $J M$ where $J = \begin{bmatrix} 0 & \mathbb{I}_{n} \\ -\mathbb{I}_{n} & 0\end{bmatrix}$ is the standard symplectic form \citep{JM}. Writing $A = S_{V} BS_{H}^{T}$ with $B \in Sp(2n)$ we find that 
\begin{equation}\label{eqn10}
\text{tr}(VAHA^{T}) = \text{tr}(B^{T}(\mathcal{D}_{V} \oplus \mathcal{D}_{V}) B(\mathcal{D}_{H} \oplus \mathcal{D}_{H})). 
\end{equation}

If $B$ is chosen to be the symplectic linear map relabeling the canonically conjugate pairs by the formula $x_{i} \mapsto x_{n+1-i}$ and $p_{i} \mapsto p_{n+1-i}$, we obtain the inequality
\begin{equation}\label{optimal}
    \inf_{A \in Sp(2n)} \text{tr}(VAHA^{T}) \leq 2 \sum_{i=1}^{n} \lambda^{H}_{i} \lambda^{V}_{n+1-i}. 
\end{equation}
It is a nontrivial fact, which we hold off on proving until Theorem \ref{main}, that the inequality in Eq.\,(\ref{optimal}) can be replaced by equality. Hence 
\begin{equation}\label{gromov}
    E_{Sp(2n)} = N V_{0} + 2\sum_{i=1}^{n} \lambda^{H}_{i} \lambda^{V}_{n+1-i}. 
\end{equation}
Just as with the linear Gardner energy, there is a continuous family of $Sp(2n)$ matrices minimizing the linear Gromov energy. Notably, for any $n$ tuple of angles $(\theta_{1},\hdots,\theta_{n})$ we can define the symplectic rotation matrix
\begin{equation} X(\theta_{i}): \begin{bmatrix}x_{i} \\ p_{i} \end{bmatrix} \mapsto \begin{bmatrix} \cos(\theta_{i}) & \sin(\theta_{i}) \\ -\sin(\theta_{i}) & \cos(\theta_i) \end{bmatrix}\begin{bmatrix}x_{i} \\ p_{i} \end{bmatrix}. \end{equation}
Replacing $B$ with $X(\theta_{i}) B$ in Eq.\,(\ref{eqn10}) leaves the trace invariant so there is a $SO(2)^{n}$ family of linear maps taking $f_{0}$ to a minimal energy configuration. This family of matrices is enlarged when the symplectic eigenvalues of either $V$ or $H$ are degenerate.

It is interesting to compare Eq.\,(\ref{gardner}) and Equation Eq.\,(\ref{gromov}). Notably, the AM-GM inequality shows that $E_{Sp(2n)} \geq E_{SL(2n)}$ since
\begin{equation}\label{comparison}
 2\sum_{i=1}^{n} \lambda^{H}_{i} \lambda^{V}_{n+1-i} \geq 2n \, (\det(HV))^{1/2n}, 
 \end{equation} 
 with equality iff $\lambda_{i}^{H} \lambda_{n+1-i}^{V} \equiv \text{const}$. In particular,  $E_{Sp(2)} = E_{SL(2)}$ as expected. 

\subsection{Example 1}\label{ex1}
 We first compute an easy example in $n = 2$. For $\epsilon > 0$, suppose that $\mathcal{E}(\mathbf{z},\epsilon) = x^2 + \epsilon^2 y^2 + p_{x}^2 + p_{y}^2$. Suppose further that $f_{0}(\mathbf{z},R) = \frac{6}{R^2 |B(R)|} \chi_{B(R)} = \frac{6}{R^2|B(R)|} \Theta(R^2-|\mathbf{z}|^2)$ is a rescaled indicator function on the ball. 
 We compute that $N=1$, $\mathbf{c} = 0$, $H(R) =  \mathbb{I}_4$, $V_{0} = 0$, $\mathbf{d} = 0$, and $V(\epsilon) = \text{diag}(1,\epsilon^2,1,1)$. The initial energy stored in $f_{0}$ is $E[f_{0}] = (3 + \epsilon) $. For $\epsilon$ small, we should expect that $E_{SL(4)}$ is small since $f_{0}$ can be squeezed onto the $y$ axis via area-preserving maps. Indeed, Eq.\,(\ref{gardner}) gives us that 
 \begin{equation}\label{example}
     E_{SL(4)} = \epsilon^{1/2}.
 \end{equation}
 Dividing this equation by $E[f_{0}]$, we can alternatively compute the inaccessible energy fraction $F_{SL(2n)} \coloneqq \frac{E_{SL(2n)}}{E[f_{0}]}$ to be
  \begin{equation}\label{example}
     F_{SL(4)} = \frac{4 \epsilon^{1/2}}{3 + \epsilon^2}.
 \end{equation}
The distribution function after an energy minimizing linear map is $f_{0}\circ \phi^{-1} = \frac{1}{|B(R)|} \Theta(R^2-\epsilon^{-1/2}(x^2+p_{x}^2+p_{y}^2)-\epsilon^{3/2}y^2)$ which looks as expected. 

In contrast, the linear Gromov's nonsqueezing theorem prohibits squeezing $f_{0}$ onto the $y$-axis via linear symplectomorphisms. This implies we should find $E_{Sp(2n)}$ to be finite in the limit $\epsilon \to 0^{+}$. $V(\epsilon)$ can be symplectically diagonalized by $S_{V} = \text{diag}( 1, \epsilon^{-1/2}, 1, \epsilon^{1/2})$ giving $\mathcal{D}_{V} = \text{diag}(1,\epsilon)$. From Eq.\,(\ref{gromov}) we can then compute
\begin{equation}\label{equation 2}
    E_{Sp(4)} = \frac{1}{2} \left(1+ \epsilon \right), 
\end{equation}
which limits to a finite value as expected. We compute the inaccessible energy fraction $F_{Sp(2n)} \coloneqq \frac{E_{Sp(2n)}}{E[f_{0}]}$ as
\begin{equation}
F_{Sp(4)} = \frac{2 + 2 \epsilon}{3 + \epsilon^2}.
\end{equation}
The distribution function after the energy-minimizing mapping is $f_{0} \circ \phi^{-1} =  \frac{6}{R^2|B(R)|} \Theta(R^2-\epsilon y^2 - \epsilon^{-1} p_{y}^2 - x^2 - p_{x}^2)$ which looks as expected since one cannot symplectically compress $f_{0}$ along the $x$ axis without increasing $p_{x}$ by a corresponding factor. This addresses a question put forth in \cite[p.4]{Qin2025}.
\begin{center}
    \begin{figure}[h!]
        \centering
        \includegraphics[width=\linewidth]{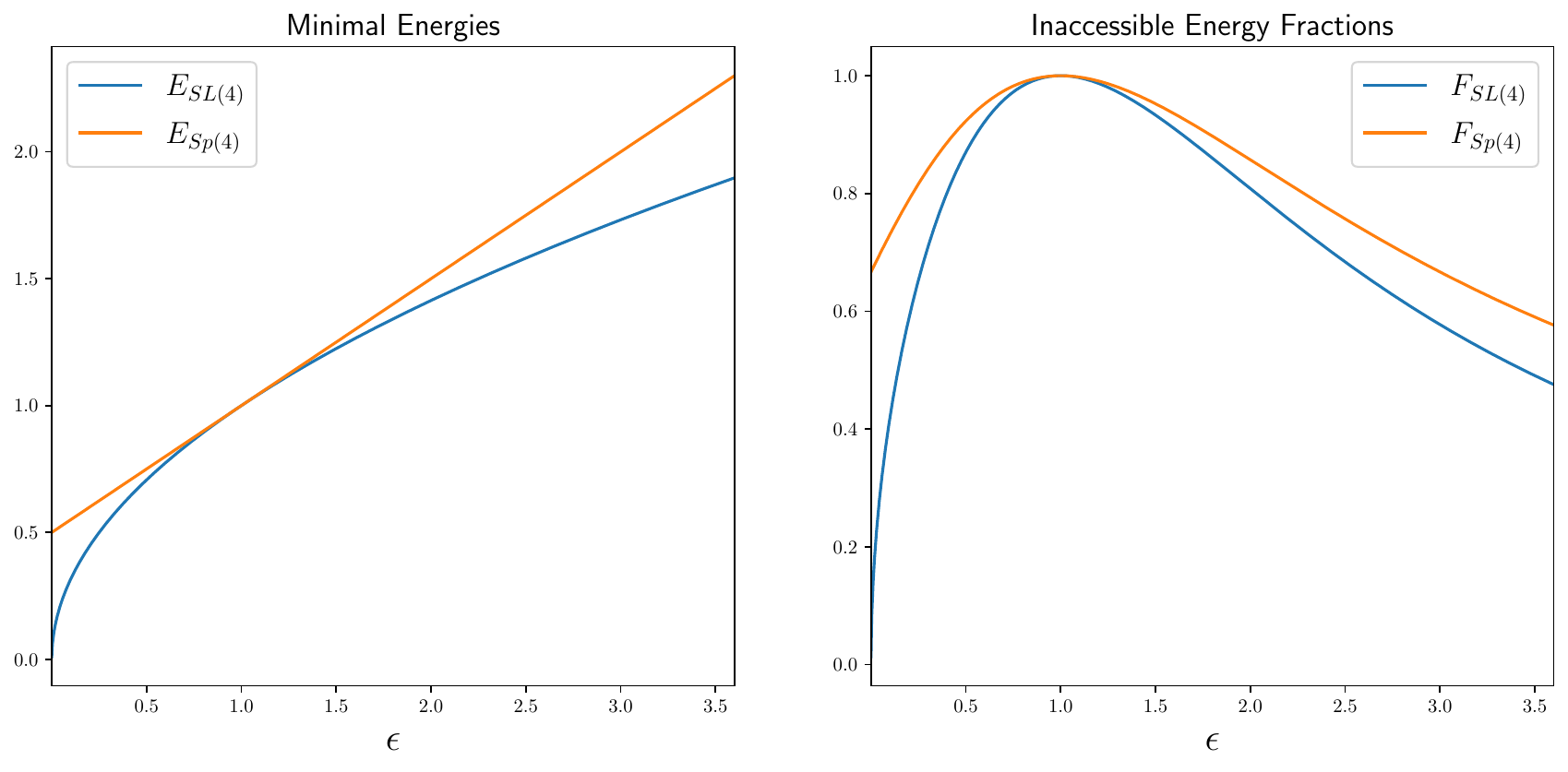}
        \caption{Minimal energies for Example \ref{ex1}.}
        \label{fig:enter-label}
    \end{figure}
\end{center}
It is interesting to note that $E_{Sp(4)}$ and $E_{SL(4)}$ can differ, either in difference or ratio, by an arbitrarily large amount. As we show in \ref{ex3}, Eq.\,(\ref{example}) is also the correct formula for the Gardner energy. This refutes a conjecture in \citep[p.4]{Qin2025} that the linear Gromov energy must be close to the (nonlinear) Gromov energy, $\inf_{\phi \in \text{Ham}(\mathcal{P})}E[f_{0} \circ \phi^{-1}]$.  
\FloatBarrier
\subsection{Example 2}\label{ex2}
Suppose we continue to consider the energy function $\mathcal{E}(\mathbf{z},\epsilon) = x^2 + \epsilon^2 y^2 + p_{x}^2 + p_{y}^2$ but $f_{0}$ is instead given by $f_{0}(\mathbf{z}) = \frac{1}{2(2\pi)^{2}} \exp(-\frac{1}{2}(p_{x}^2+p_{y}^2 + \frac{1}{4} x^2 + y^2))$. Then either by direct integration or by computing the Fourier transform of $f_{0}$, we learn that $H = \text{diag}(4, 1,1,1)$. $H$ is symplectically diagonalized by $S_{H} = \text{diag}(2^{-1/2}, 1, 2^{1/2},1)$, implying $D_{H} = \text{diag}(2,1)$. The linear Gardner energy is therefore 
\begin{equation}
    E_{SL(4)}(\epsilon) = 4 \sqrt{2\epsilon},
\end{equation}
whereas 
\begin{equation} \label{dis}
     E_{Sp(2n)}(\epsilon) =\begin{cases} 
      4 \epsilon + 2 & 0 < \epsilon < 1 \\
      2 \epsilon +4 & \epsilon \geq 1.
   \end{cases}
   .
\end{equation}
\begin{center}
    \begin{figure}[h!]
        \centering
        \includegraphics[width=\linewidth]{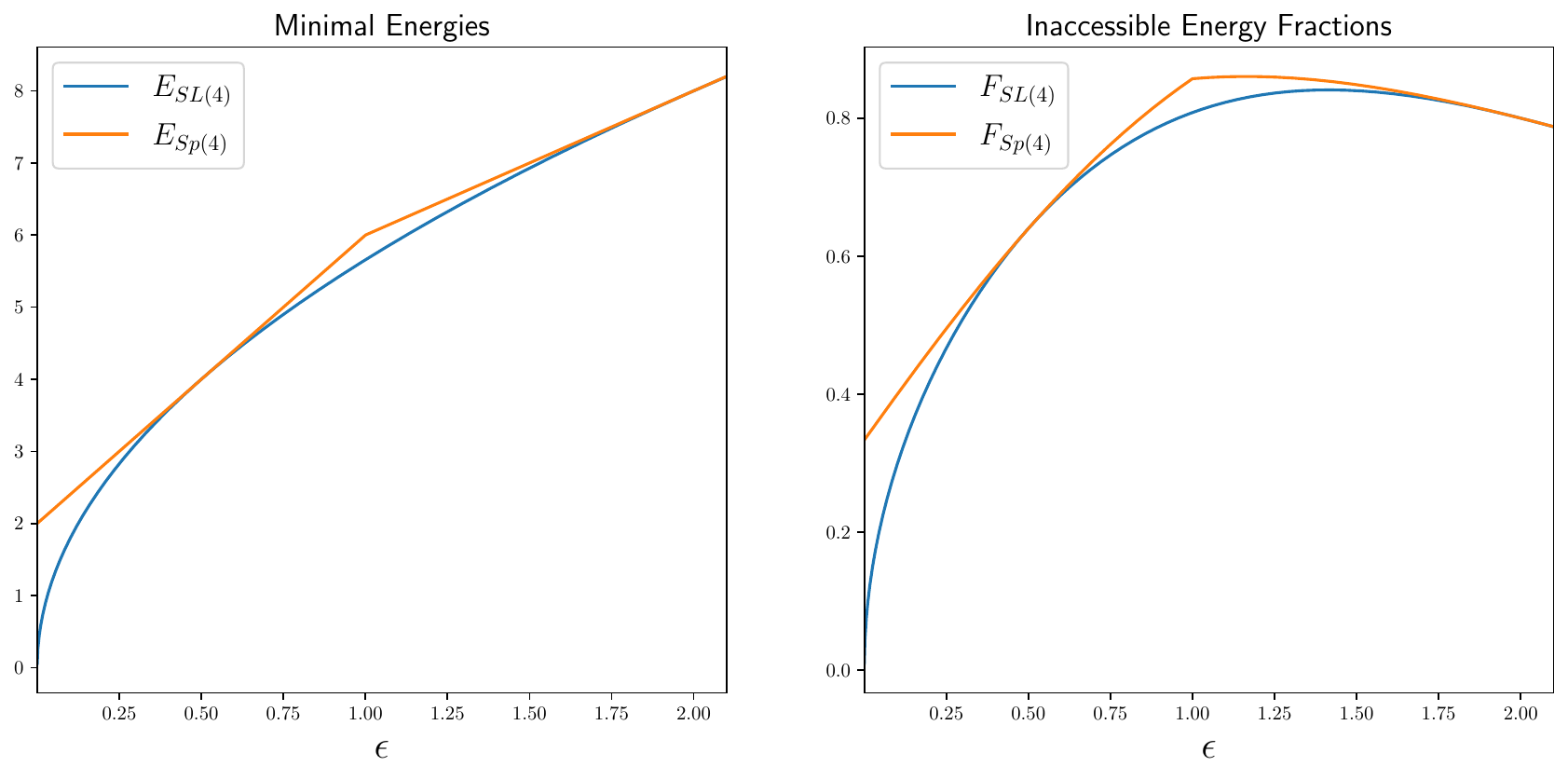}
        \caption{Minimal energies for Example \ref{ex2}.}
        \label{fig:enter-label}
    \end{figure}
\end{center}
\FloatBarrier
Equation (\ref{dis}) illustrates that even for smoothly varying $H$ and $V$, the linear Gromov energy does not need to vary smoothly. A lack of differentiability of $E_{Sp(2n)}(\epsilon)$ at points where either $D_{V}$ or $D_{H}$ has repeated symplectic eigenvalues is to be expected since the symplectic eigenvalue pairing in Eq.\;(\ref{gromov}) is generically reordered. At these points, a sort of saturation occurs and a "large" symplectomorphism exchanging canonical pairs must be applied before further energy can be extracted from $f_{0}$ via smoothly varying symplectomorphisms. The lack of such non-differentiable behavior in the linear Gardner energy exemplifies the flexibility of area-preserving maps in comparison to symplectic maps.
\ignorespacesafterend
\subsection{Example 3: Symplectic Equivalence of Ellipsoids}  \label{ex3}
We now consider a subclass of problems for which the Gardner and linear Gardner energy agree. We will use our observations to prove that two ellipsoids are linearly symplectomorphic iff their defining matrices have the same symplectic eigenvalues.

Given $2n \times 2n$ symmetric, positive-definite matrix $M$ we define the ellipsoid $El(M) \coloneqq \{ \mathbf{z} \in \mathbb{R}^{2n}: \mathbf{z}^{T} M \mathbf{z} \leq 1\}$. For $M$ and $M'$ fixed symmetric, positive-definite matrices with the same determinant, we take $\mathcal{E}(\mathbf{z}) = \mathbf{z}^{T} M \mathbf{z}$ and $f_{0} = \chi_{El(M')}$. Since $\det(M) = \det(M')$ the areas of $El(M')$ and $El(M)$ agree. It is easy to check that there is a $SL(2n)$ matrix $A$ mapping $El(M')$ to $El(M)$. It is also not hard to see that $E_{G} = E_{SL(2n)} = \int \mathcal{E}(\mathbf{z}) f_{0}(A^{-1} \mathbf{z})$. Indeed, suppose $\phi: \mathbb{R}^{2n} \to \mathbb{R}^{2n}$ is any area-preserving diffeomorphism not mapping $El(M')$ to $El(M)$. Then part of $\phi(El(M'))$ lies outside of $El(M)$, say the set $S$. Since $\mathcal{E}|_{S} > 1$, more energy could be extracted from $f_{0}$ by moving $S$ into $El(M)$. Given that the volume of $El(M)$ and $El(M')$ agree, this proves that the optimal amount of energy from $f_{0}$ is extracted by mapping $El(M')$ to $El(M)$. Conversely, this argument shows that if a map $\phi$ minimizes $E[f_{0} \circ \phi^{-1}]$ then $\phi$ necessarily maps $El(M')$ to $El(M)$. 

By our observations, we conclude that $El(M)$ and $El(M')$ are linearly symplectomorphic iff both $\det(M) = \det(M')$ and $E_{Sp(2n)} = E_{SL(2n)}$. For the given $f_{0}$ and $\mathcal{E}$, we have that $V_{0} = 0$, $\mathbf{c} = 0$, $V=M$, and $H  \propto (M')^{-1}$. The symplectic eigenvalues of $H$ are proportional to the reciprocals of the symplectic eigenvalues of $M'$ so, with a uniform proportionality constant, $\lambda_{i}^{H} \propto \frac{1}{\lambda^{M'}_{n+1-i}}$. Hence by equality condition for Eq.\,(\ref{comparison}), we learn that $El(M)$ and $El(M')$ are linearly symplectomorphic iff $\frac{\lambda_{i}^{M}}{\lambda_{i}^{M'}} \equiv \text{const}$ and $\det(M) = \det(M')$. These two conditions are combine to give $\lambda_{i}^{M} = \lambda_{i}^{M'}$, concluding the proof.
\subsection{Example 4: Gromov's non-Squeezing Theorem}
In this last example, we prove the affine Gromov nonsqueezing theorem using our energy-minimization theory. Namely, we will show for $R>r$ there is no affine symplectomorphism taking the ball $B(R)$ into the symplectic cylinder $\mathcal{Z}(r)$.

We define the energy $\mathcal{E}(\mathbf{z}) = x_{1}^2 + p_{1}^2$ so $V = e_{x_{1}}e_{x_{1}} + e_{p_{1}}e_{p_{1}}$. For $R > r$ we take $f_{0} =\chi_{B(R)}$. We compute $H = \frac{R^{2n+2}}{2n+2(2n)}|S^{2n-1}|\mathbb{I}_{2n}$. Using the result of \ref{appendix} to justify the limit swap $\lim_{\epsilon \to 0^{+}} \inf_{X \in Sp(2n)} \text{tr}(X (V + \epsilon \mathbb{I}_{2n}) X^{T} H) = \inf_{X \in Sp(2n)}\text{tr}(X V X^{T} H)$, we compute $E_{Sp(2n)} = \frac{R^{2n+2}}{n+1(2n)}|S^{2n-1}| = \int f_{0}(\mathbf{z})\mathcal{E}(\mathbf{z})d^{2n}\mathbf{z}$. This shows there is no way to reduce the energy of $f_{0}$ by symplectomorphisms. Under any affine symplectomorphism, $f_{0}$ transforms into the characteristic function of an ellipsoid. We will show for every ellipsoid $El \subset \mathcal{Z}(r)$ linearly symplectomorphic to $B(R)$ that $\int_{El} \mathcal{E} < E_{Sp(2n)}$, proving no such $El$ exists.

If $\mathbf{z}_{0} + El(M^{-1})$ is an ellipsoid in $\mathcal{Z}(r)$ then we can move $\mathbf{z}_{0} + El(M^{-1})$ to the coordinate origin without increasing its energy and while keeping it confined to $\mathcal{Z}(r)$. Without loss of generality, we thereby set $\mathbf{z}_{0} = 0$. We compute that
\begin{equation}
\int_{El(M^{-1})} x_{1}^2 + p_{1}^2 = \sqrt{\det(M)}\frac{|S^{2n-1}|}{2n(2n+2)} (M_{1,1} + M_{n+1,n+1}). 
\end{equation}
By assumption, $El(M^{-1}) \subset \mathcal{Z}(r)$, so for every unit vector $\mathbf{v}$ we have $\left< x_{1}, M^{1/2} \mathbf{v} \right>^2 + \left< p_{1}, M^{1/2} \mathbf{v}\right>^2 \leq r^2$. Since $M^{1/2}$ is self-adjoint, taking $\mathbf{v} =  ||M^{1/2}x_{1}||^{-1} M^{1/2}x_{1}$ shows that $||M^{1/2}x_{1}||^2 \leq r^2$. Similarly, we conclude $ ||M^{1/2} p_{1}||^2 \leq r^2$. These inequalities imply that $M_{1,1} \leq r^2$ and $M_{n+1,n+1} \leq r^2$. If, for the sake of contradiction, $B(R)$ is linearly symplectomorphic to the ellipsoid $El(M^{-1})$ then 
\begin{equation}
\frac{R^{2n+2}}{2n(n+1)}|S^{2n-1}| \leq \frac{r^2 \sqrt{|M|}|S^{2n-1}|}{2n(n+1)}. 
\end{equation}
Since the areas of $B(R)$ and $El(M)$ necessarily agee, it must be that $|M| = R^{4n}$. Simplifying, we have the erroneous inequality $R^2 \leq r^2$, the desired contradiction. 
 \begin{theorem}\label{main}
    For any $2n \times 2n$, symmetric, positive-definite matrices $V,H$ with respective symplectic eigenvalues $\lambda_{1}^{V} \geq \hdots \geq \lambda_{2n}^{V} > 0$ and $\lambda_{1}^{H} \geq \hdots \geq \lambda_{2n}^{H} > 0$, 
    \[
    \inf_{S \in Sp(2n)} \text{tr}(S V S^{T} H) = 2 \sum_{i=1}^{n} \lambda_{i}^{H} \lambda^V_{n+i-1}.
    \]
\end{theorem}
 \begin{proof} To complete the proof of Theorem $\ref{main}$, we have only to show that 
 \begin{equation}\label{18}
 \inf_{S \in Sp(2n)}\text{tr}(S^{T}(\mathcal{D}_{V} \oplus \mathcal{D}_{V}) S(\mathcal{D}_{H} \oplus \mathcal{D}_{H})) \geq 2 \sum_{i=1}^{n} \lambda_{i}^{H} \lambda_{n+1-i}^{V}. 
 \end{equation}
 We will prove Eq.\,(\ref{18}) with a method similar to Son and Stykel \cite{Son} using a theorem of Liang et al. \citep{generalizing}.
\begin{theorem}{\cite[p.489]{generalizing}}\label{GTM}
Let $A,B \in \mathbb{C}^{d \times d}$ and $D_{\pm} \in \mathbb{C}^{k_{\pm} \times k_{\pm}}$ be Hermitian matrices such that $A\neq 0$, $B$ has both positive and negative eigenvalues, $k_{+} + k_{-}= d$, $D_{\pm} \geq 0$, and the matrix pencil $A-\lambda B$ is positive semi-definite. Let
\[
J_{k} = \begin{bmatrix} \mathbb{I}_{k_{+}} & 0 \\ 0 & -\mathbb{I}_{k_{-}} \end{bmatrix}, \; \; \; D =  \begin{bmatrix} D_{+} & 0 \\ 0 & D_{-} \end{bmatrix},
\]
and let $\omega^{+}_{1} \geq \hdots \geq \omega_{k_{+}}^{+}\geq0 $ and $\omega_{1}^{-} \geq \hdots \geq \omega_{k_{-}}^{-} \geq 0$ be the eigenvalues of $D_{+}$ and $D_{-}$ respectively. Let $\lambda_{k_{-}}^{-} \leq \hdots \leq \lambda_{1}^- \leq \lambda_{1}^+ \leq \hdots \leq \lambda_{k_{+}}^+$ be the eigenvalues of the matrix pencil $A-\lambda B$. Then
\[
\inf_{X^\dagger B X=J_{k}} \text{tr}(D X^{\dagger} A X) = \sum_{i=1}^{k_{+}} \lambda_{i}^{+} \omega_{i}^{+} - \sum_{i=1}^{k_{-}}  \lambda_{i}^{-}\omega^{-}_{i}.
\]
\end{theorem}

For our purposes, we will take $d = 4n$, $k_{\pm}=2n$, $D_{+} = D_{-} = \mathcal{D}_{V} \oplus \mathcal{D}_{V}$, and $A = \mathcal{D}_{H} \oplus \mathcal{D}_{H}\oplus\mathcal{D}_{H} \oplus \mathcal{D}_{H}$. Trivially, $\omega_{2i-1}^{+} = \omega^+_{2i} = \omega_{2i-1}^{-} = \omega_{2i}^{-} = \lambda_{i}^{V}$ are the symplectic eigenvalues of $V$. As in \citep{Son}, for any $S \in Sp(2n)$ we define
\begin{equation}
    B = \begin{bmatrix}0 & J \\ -J & 0 \end{bmatrix}, \; \; \; \mathcal{X}(S) = \frac{1}{\sqrt{2}}\begin{bmatrix}
    S & SJ^{T} \\
    SJ^{T} & S 
\end{bmatrix}.
\end{equation}
$B$ is Hermitian and has eigenvalues $\pm 1$, each with multiplicity $2n$. Considering $\lambda =0$ shows the matrix pencil $A-\lambda B$ is positive definite. It's also easy to verify that the eigenvalues of the pencil $A-\lambda B$ are $\pm \lambda_{i}^{H}$, each with multiplicity two. Since $\mathcal{X}(S)^{\dagger} B \mathcal{X}(s) = J_{k}$, and since all the hypotheses of the Theorem \ref{GTM} are satisfied,
\begin{equation}\label{inequality}
\text{tr}(D \mathcal{X}^\dagger(S) A \mathcal{X}(S)) = 2 \, \text{tr}(S (\mathcal{D}_{V} \oplus \mathcal{D}_{V}) S^{T} (\mathcal{D}_{H} \oplus \mathcal{D}_{H})) \geq 4 \sum_{i=1}^{n} \lambda_{i}^{V} \lambda_{n+1-i}^{H}. 
\end{equation}
Since Eq.\,(\ref{inequality}) holds for every symplectic matrix, we conclude our proof of Theorem \ref{main}.
\end{proof}
 \section{Discussion}
 Efficiently extracting energy from particle distributions is vital for making fusion energy a reality. It is therefore critical to have useful bounds for the amount of extractable energy in a plasma. Complicating the minimization of energy in a system is the constraint that particles ideally evolve by Hamiltonian symplectomorphisms. With this constraint in mind, we were able to compute the minimal energy of a system in two extremes: one has arbitrarily fine control of phase space, and one has only very coarse control of phase space. The former energy we showed was equivalent to the Gardner energy $E_{G}$, while the latter we showed could be computed with Eq.\,(\ref{gromov}). Both of these cases were interesting, highlighting the flexibility vs. rigidity of symplectomorphisms. Nonetheless, our analysis seems unsatisfactory in places. The distribution function becomes pathological when minimizing energy using nonlinear symplectomorphisms. On the other hand, in coarse-graining the problem, we were left with transformations too blunt to efficiently extract energy from many classical non-thermal distributions (see \ref{lvsnl}). Given the rich theory already developed here, it is likely a rich medium-ground may be found. There are many non-equivalent approaches to explore, but we leave this for future work.

 Our analysis is not meant to conclude the question of energy extraction. We have not addressed thermodynamic or realistic technical constraints. Rather, we view this work as a first step towards understanding phase space engineering and the constraints of symplectomorphisms. To this end, much work remains to be done. Even within the limited problem of exacting energy with a restricted set of symplectomorphisms, many families of transforms may be fruitfully considered. Perhaps the most important family of transformations to consider are those corresponding to a Hamiltonian of the form $\mathcal{H} = \frac{1}{2}(p - q A(x,t))^2 + \phi(x,t)$, or even simply those of the form $\mathcal{H} = \frac{p^2}{2} + \phi(x,t)$. Understanding these flows would lead to a better understanding of electromagnetic phase-space engineering. We leave such a study to a future work.

While not always sharp, the linear Gromov energy gives a good benchmark for energy extraction, even when only the fluid moments of a distribution are known. The linear Gromov energy can also be useful when computing the allowed energy under non-linear symplectomorphisms. Composing nonlinear maps with optimal linear maps can give a computationally feasible way to compute the Gardner energy. In a future publication \citep{Bohlsen}, we will present a machine learning inspired numerical approach to computing upper bounds on the nonlinear Gromov energy with both free and periodic spatial boundary conditions. This approach will allow us to investigate the effect of regularizing the Gromov-Gardner problem by constraining the gradient of the symplectomorphism, a question which is functionally intractable within the linear theory presented here.

Perhaps the biggest goal of this work has been to show how physical problems and mathematical theory can fruitfully coalesce. In deriving the ground state energy of a particle distribution, we were led to prove a novel trace-minimization theorem. Many more results may yet be derived from the problem of phase-space engineering. These results would deepen our understanding of how particles can be manipulated, and might themselves be of mathematical interest. 
\section{Appendix}
\subsection{Degenerate Potentials} \label{appendix}
Suppose the potential matrix $V$ is merely positive semidefinite. Then for either $G= Sp(2n)$ or $G= SL(2n)$, we will show that $E_{G} = N V_{0} + \inf_{X \in G}\text{tr}(X V X^{T} H) = N V_{0} + \lim_{\epsilon \to 0} \inf_{X \in G}\text{tr}(X (V + \epsilon \mathbb{I}_{2n}) X^{T} H) $. Since $V+\epsilon \mathbb{I}_{2n}$ is positive definite for $\epsilon > 0$, this allows us to compute $E_{G}$ using Eq.\,($\ref{gardner}$) or Eq.\,($\ref{gromov}$). 
\begin{proof}
Define for $\epsilon \geq 0$ the nonnegative function $E_{G}(\epsilon) = NV_{0} + \inf_{X \in G} \text{tr}(X (V + \epsilon \mathbb{I}) XH)$. We must show $E(\epsilon)$ is continuous at $\epsilon = 0$. It is easy to verify that $E(\epsilon)$ is monotone increasing and in particular that $E(0) \leq E(\epsilon)$ for any $\epsilon \geq 0$. Let $\delta > 0$ be fixed. Let $X_{0}$ be such that $\text{tr}(X_{0}VX_{0}H) - E(0) < \eta/2$. Then for any $\epsilon >0$ such that $\epsilon \text{tr}(X_{0}VX_{0} H) < \delta/2$, the triangle inequality implies 
\begin{equation}
\delta > \text{tr}(X_{0}(V+\epsilon \mathbb{I}_{2n})X_{0}H) - E(0) \geq E(\epsilon) - E(0) \geq 0.
\end{equation}
Hence $\lim_{\epsilon \to 0} E(\epsilon) = E(0)$, proving the result. 
\end{proof}
\subsection{Gardner's Restacking and Hamiltonian Approximations}\label{Gromov=Gardner}
We elaborate on Gardner's restacking algorithm and Hamiltonian approximations thereof. For simplicity, we will assume $\mathcal{P} = \mathbb{R}^{2n}$, but there is no obstruction in allowing $\mathcal{P}$ to be a manifold. We will also make the simplifying assumption that $f_{0}$ is continuous, but this restriction can easily be lifted. As always, we assume $\mathcal{E}$ is bounded below and that $f_{0} \mathcal{E}$ is integrable.

When $\{\phi\}$ are required to be area-preserving, but possibly discontinuous, Gardner's restacking allows us to compute $E_{G} = \inf_{\{\phi \}} E[f_{0} \circ \phi^{-1}]$. To describe the algorithm in a mathematically rigorous manner, let $n \in \mathbb{N}$ be a natural number. Let $h = 2^{-n}$. We define the lattice $\Lambda^{h} = h \mathbb{Z}^{2n}$. For $\lambda \in \Lambda^h$, we define $S^h_{\lambda} = \lambda + [0,h]^{2n}$. We break phase space into a collection of disjoint squares of side length $h$, viz $\mathcal{P} = \bigsqcup_{\lambda \in \Lambda^h} S^h_{\lambda}$. 

Since $\Lambda^h$ is countable, we may index every element $\lambda \in \Lambda^h$ by a natural number, denoted by $\lambda_{i}$. We choose two such indexings. The first denoted by a subscript $f$ is such that $f_{0}(\lambda_{1,f}) \geq f_{0}(\lambda_{2,f}) \geq \hdots \geq 0$. The second indexing of $\Lambda^h$, denoted with a subscript $\mathcal{E}$, is such that $\mathcal{E}(\lambda_{1,\mathcal{E}}) \leq \mathcal{E} (\lambda_{2,\mathcal{E}}) \leq \hdots$.
We then define an bijective, area-preserving map $\phi^{(n)}$ sending $S^h_{\lambda_{i,f}} \mapsto S^h_{\lambda_{i,\mathcal{E}}}$, say the identity map on squares. Essentially, $\phi^{(n)}$ permutes equal-area sets in phase space until we have minimized the energy of the approximate distribution function $f_{0}^{(n)} = \sum_{\lambda \in \Lambda^{h}} f_{0}(\lambda^{h}) \chi_{S^h_{\lambda}}$. The ground state energy of $f^{(n)}_{0}$ is computed to be $E^{(n)} = \sum_{i=1}^{\infty} f(\lambda_{i,f})\mathcal{E}(\lambda_{i,\mathcal{E}})$. $E_{G}$ is computed as the limit of $E^{(n)}$ as $n$ tends towards infinity. 

We note that the distribution function after the energy minimizing mapping, formally denoted $f_{\infty} = \lim_{n \to \infty} f^{(n)} \circ \phi^{(n)}$, may or may not uniquely exist. In $\mathcal{E}$ has a level set of non-zero measure, then one generically has an infinite number of possible final states, $f_{\infty}$. If $\mathcal{E}$ is nondegenerate except around its minimum, then it is likely that $f_{\infty}$ is a well-posed measure. In 1D for example, if $f_{0}(p)$ and $\mathcal{E}(p) = \frac{p^2}{2}$, then $f_{\infty}(p)$ is the symmetric decreasing rearrangement of $f_{0}$ \citep{book, Day_1972}. 

To show that $\inf_{\phi \in \text{Ham}(\mathcal{P})} E[f_{0} \circ \phi^{-1}] = E_{G}$, we show that it is possible to approximate the steps in Gardner's restacking with Hamiltonian maps. 

For convenience, we will assume that $f_{0}$ is compactly supported. Let $h$ be as before. Let $R$ be a square containing the support of $f_{0}$. We can apply Gardner's restacking to $R$ as described above. Let $S_{1},\hdots,S_{N}$ be squares of side length $h$ covering $R$. Gardner's algorithm permutes these squares, mapping $S_{i} \mapsto S_{\sigma(i)}$ in some area-preserving manner. In order to so the same with Hamiltonian maps, let $\delta > 0$. We remove a $\frac{\delta}{2N}$ amount of area from the edges of each square $S_{i}$ making each trimmed square $\tilde{S}_{i}$ closed and disjoint. Then Katok \citep[p.545]{Katok} shows there is a Hamiltonian diffeomorphism $\psi^{(n)}$ supported in a neighborhood of $R$ almost mapping $\tilde{S}_{i} \mapsto \tilde{S}_{\sigma(i)}$ in the precise sense that $|(\psi^{(n)})^{-1}(\tilde{S}_{\sigma(i)}) \cap \tilde{S}_{i}| < \frac{\delta}{2N}$. As we send $\delta \to 0$, we closely approximate a map sending $S_{i}$ to $S_{\sigma(i)}$. The region of measure $\delta$ not mapped according to Gardner's algorithm contributes an error of at most $\delta ||\mathcal{E}||_{L^\infty(R)} ||f_{0}||_{L^\infty(R)}$ to the energy integral. Taking $\delta= \exp(1/h)$ and taking the limit $h \to 0$ therefore shows that $E_{G} = \inf_{\phi \in \text{Ham}(\mathcal{P})} E[f_{0} \circ \phi^{-1}]$. 

It is interesting to compare this result to that of Kolmes and Fisch \citep{Kolmes2020}. They showed that Gardner's restacking could be arbitrarily approximated by diffusive operations. Given that many fundamental processes can lead to apparent diffusion, it is perhaps not surprising that Gardner's restacking can be approximated by Hamiltonian diffeomorphisms. 
\subsection{Linear vs. Nonlinear Operations}\label{lvsnl}
While the "bluntness" of linear maps avoided the problem of over-controlling phase space, this same bluntness prevents energy from being extracted in many classical scenarios. Consider, for example, a 1D idealized bump on tail distribution 
\[
f_{0}(p) = \frac{n_{0}}{\sqrt{2 \pi T}} \exp(-p^2/2T) + n_1 \delta(p - p_{0}) = f_{\text{eq}} + f_{\text{bump}},
\]
where $T, n_{1}$ and $n_{0}$ are constants. The energy function is the classical energy $\mathcal{E} = \frac{p^2}{2}$. $f_{0}$ is spatially homogeneous, but it still makes sense to speak of the energy density. Using R-F waves, one can cause $f_{0}$ to form a quasi-linear plateau around $p = p_{0}$, extracting energy in the process. Optimally, one could move $f_{\text{bump}}$ to $p = 0$ by Gardner restacking. The resulting distribution function has the lowest possible energy since $f_{eq}$ is symmetric and decreasing. This extracts an energy per unit volume of $p_{0}^2 n_{1}$ giving an Gardner energy density of $\frac{n_{0}T}{2}$. 

With linear operations, none of the previous operations are allowed. The only $x$ independent, affine, area-preserving map is the shift map $p \mapsto p + \delta p$. To extract energy from $f_{0}$, all we can do is shift $f_{0}$ until there is no net momentum i.e. $\int pf_{0}(p- \delta p) dp = 0$. The energy density of $f_{0}(p - \delta p)$ is 
\begin{equation}
    \frac{n_{0}}{2}(\delta p^2 + T) + \frac{n_{1}(p_{0} + \delta p)^2}{2}, 
\end{equation}
which is minimized when $\delta p =- \frac{n_{1} p_{1}}{n_{0} + n_{1}}$. In this case, the energy extractable by linear maps is a pittance compared to the energy extractable by nonlinear maps.

We can play similar games in multiple dimensions with slightly more interesting results. For example, suppose we have a 3D bump on-tail
\[
f_{0}(\mathbf{p}) = n_{0}\left(\frac{1}{2 \pi T}\right)^{3/2} \exp(-p^2/2T) + n_1 \delta(\mathbf{p} - \mathbf{p}_{0}) = f_{\text{eq}} + f_{\text{bump}}.
\]
Then the Gardner energy is again obtained by moving $f_{\text{bump}}$ to $\mathbf{p} = 0$. Restricting to linear maps, we can again shift momentum space in the direction of $\mathbf{p}_{0}$ space until $\int \mathbf{p} f_{0}(\mathbf{p} - \delta \mathbf{p}_{0}) d^{3}\mathbf{p} = 0$. Now with more dimensions, we are additionally able to squeeze momentum space in the $\mathbf{p}_{0}$ direction while uniformly expanding the orthogonal plane. While still far from saturating the nonlinear energy bound, the linear energy is nonetheless closer due to the increased flexibility of higher dimensions.  

Area-preserving linear maps cannot alter the internal structure of $f_{0}$. It is no surprise that the linear Gardner and Gromov energies often depend only on the first few moments of $f_{0}$. This is both a blessing and a curse. At least in many cases, $E_{Sp(2n)}$ and $E_{SL(2n)}$ can be computed directly from a fluid theory. No information about the kinetics is needed. The linear energies can therefore serve as a useful upper bound for the nonlinear energies in the absence of an exactly known distribution $f_{0}$. However, when $f_{0}$ is known to some accuracy, and manipulations can be performed on scales smaller than the support of $f_{0}$, the linear theory may fail to yield useful results. 
\subsection{Quadratic Hamiltonian and Linear Maps}\label{Ham}
To understand how linear maps may arise naturally, consider Hamiltonian's equations on $\mathbb{R}^{2n}$,
\begin{equation}
\dot{\mathbf{z}} = J \nabla \mathcal{H}(\mathbf{z}). 
\end{equation}
If, at least to an approximation, $\mathcal{H}(z,t) =\frac{1}{2} \mathbf{z}^{t}\mathbb{H}\mathbf{z}$ is a quadratic polynomial then 
\begin{equation}\label{Ham flow}
\dot{\mathbf{z}} = J \mathbb{H} \mathbf{z}. 
\end{equation}
The flowmap $\phi_{t}$ is easily computed as $\phi_{t}(\mathbf{z}_{0}) = \exp(J\mathbb{H}t) \mathbf{z}_{0}$. The matrix $\exp(J\mathbb{H}t) $ is symplectic, and thus $\phi_{t}$ is of the allowed linear form. Alone, however, matrices of the form $\exp(J\mathbb{H}t)$ do not give the full symplectic group. To get $Sp(2n)$, we must consider time-dependent Hamiltonian flows. If $\mathbb{H}$ is allowed to be time dependent then the formal solution to Eq.\,(\ref{Ham flow}) is $\phi_{t}(\mathbf{z}_{0}) = T\exp(\int_{0}^{t} \mathbb{H}(s)ds)\mathbf{z}_{0}$, where $T\exp$ is the time-ordered exponential. It can be shown that matrices of the form $T\exp(J \int_{0}^{t} \mathbb{H}(s)ds)$ are precisely the symplectic matrices. In fact, only using that $Sp(2n)$ is connected, it was shown by W\" ustner that every symplectic matrix can be written as a product of the form  $\exp(J \mathbb{H}_{1}) \exp(J \mathbb{H}_{2})$ \citep{Wustner2003}. 

Affine linear maps are achieved by additionally considering Hamiltonians of the form $\mathcal{H} = -J \mathbf{b} \cdot z$, in which case Hamilton's equations read 
\begin{equation}
\dot{\mathbf{z}} = \mathbf{b}.
\end{equation}
Trivially, $\phi_{t}(\mathbf{z}_{0}) = \mathbf{z}_{0} + \mathbf{b}t$. Composing this flow with a linear flow gives all the affine linear flows. Thus affine symplectomorphisms arise naturally when the Hamiltonian of a system is well-approximated by a quadratic polynomial. 
 \begin{acknowledgments} 
This research was supported by the U.S. Department of Energy (DE-AC02-09CH11466)
and, in part, by DOE Grant No. DE-SC0016072.
\end{acknowledgments}
\bibliographystyle{apsrev4-2}
\bibliography{refs}
\end{document}